\documentclass[10pt,conference]{IEEEtran}

\usepackage{cite}
\usepackage{graphicx,color,epsfig,rotating}
\usepackage{amsfonts,amsmath,amssymb}
\usepackage{algorithm,algorithmic}
\usepackage{subfig}

\long\def\comment#1{}

\newfont{\bbb}{msbm10 scaled 700}

\newfont{\bb}{msbm10 scaled 1100}

\newtheorem{definition}{Definition}
\newtheorem{theorem}{Theorem}
\newtheorem{lemma}{Lemma}

\newtheorem{proof}{Proof}

\begin{document}

\title{Bounding Multiple Unicasts through Index Coding and Locally Repairable Codes}
\author{
\IEEEauthorblockN{Karthikeyan Shanmugam and  Alexandros G. Dimakis } 
\IEEEauthorblockA{ Department of Electrical and Computer Engineering \\
University of Texas at Austin \\
Austin, TX 78712-1684 \\
\texttt{karthiksh@utexas.edu,dimakis@austin.utexas.edu}}
}

\date{\today}

\maketitle

\begin{abstract}
         We establish a duality result between linear index coding and Locally Repairable Codes (LRCs). Specifically, we show that a natural extension of LRCs we call Generalized Locally Repairable Codes (GLCRs) are exactly dual to linear index codes. In a GLRC, every node is decodable from a specific set of other nodes and these sets induce a recoverability directed graph. We show that the dual linear subspace of a GLRC is a solution to an index coding instance where the side information graph is this GLRC recoverability graph. We show that the GLRC rate is equivalent to the complementary index coding rate, \textit{i.e.} the number of transmissions saved by coding. 

Our second result uses this duality to establish a new upper bound for the multiple unicast network coding problem. In multiple unicast network coding, we are given a directed acyclic graph and $r$ sources that want to send independent messages to $r$ corresponding destinations. Our new upper bound is efficiently computable and relies on a strong approximation result for complementary index coding. We believe that our bound could lead to a logarithmic approximation factor for multiple unicast network coding if a plausible connection we state is verified. \end{abstract}

\section{Introduction}

Index coding is a stylized noiseless broadcasting problem with receiver side information. 
It is extremely simple to describe and was introduced by Birk and Kol~\cite{birk1998informed} motivated by a satellite broadcasting application. Despite this initial simplicity, the problem been proven tremendously challenging and theoretically deep. 
Bar-Yossef \textit{et al.}~\cite{bar2006index} studied the problem graph theoretically where it was shown that the scalar linear optimal solution is related to a rank minimization problem over a finite field. It turns out that (for a given field size), scalar linear index coding is equivalent to a graph theoretic quantity $\mathrm{minrank}$ introduced by 
Haemers~\cite{haemers1978upper} in 1978 to obtain a bound for the Shannon graph capacity~\cite{shannon1956zero}. It is known that finding the length of the optimal scalar linear index code is computationally intractable to find and hard to approximate within a constant factor \cite{langberg2011hardness,peeters1996orthogonal}. 

Interest in index coding is further increasing due to two recent developments: 
The first is that it was recently shown~\cite{el2010index,effros2012equivalence} that any arbitrary 
network coding problem with potentially multiple sources and receivers can be mapped to a properly constructed index coding instance. Therefore, statements about index coding can be translated to constructions or bounds for general networks, showing the surprising expressiveness of the problem. 
Second, deep connections between interference alignment and index coding are being discovered 
\cite{bar2011index}\cite{maleki2012index}\cite{jafar2013topological} bringing an arsenal of new techniques for index code constructions. Further, there have been information theoretic approaches to this problem \cite{Kim2013}\cite{unal2013}.

\subsection{Our Contributions:} 
\noindent We establish two main results: The first is a duality between linear index coding and Locally Repairable Codes (LRCs)
\footnote{At the time of submission, we became aware of a concurrent independent work by Mazumdar \cite{arya} establishing similar results. Our work establishes that for vector linear codes,
the dual code (e.g. linear null space) of a GLRC is a valid index code and vice versa. Mazumdar \cite{arya} discusses a more general case of non linear codes. For that case, in one direction, \cite{a rya} shows that existence of a $k$ dimensional GLRC implies the existence of a $n-k+f(n,k,q)$ dimensional index code where the function $f$ can be found in \cite{arya} and $q$ is the field size used. Our result does not have the $f(n,k,q)$ gap term but only applies to vector linear codes.}. 
Locally repairable codes were recently developed~\cite{oggier2011self,gopalan2012locality,prakash2012optimal,papailiopoulos2012simple,papailiopoulos2012locally} to simplify repair problems for distributed storage systems and are currently used in production~\cite{huang2012erasure}. Here, we show that a natural extension that we call Generalized Locally Repairable Codes (GLCRs) are exactly dual to linear index codes. 
Specifically, in a GLRC, every node is decodable from a specific \textit{recoverability} set of other nodes. These specifications induce a \textit{recoverability directed graph}. We show that the dual linear subspace of a GLRC is a solution to an index coding instance where the side information graph is taken to be the recoverability graph of the GLRC. Therefore, the rate of the GLRC is the redundancy of the index code. The redundancy of the index code is called the \textit{complementary index coding} rate in the literature~\cite{compindex}. This quantity is the number of transmissions \textit{saved} in the index coding problem. Our proof relies on simple linear algebra and gives a clear connection between code locality and index coding. 
 
Our second result uses this duality to establish a new upper bound for the multiple unicast network coding problem. In multiple unicast network coding, we are given a directed acyclic graph and $r$ sources that want to send independent messages to $r$ corresponding destinations. It is one of the most fundamental network coding problems and has been extensively studied (\textit{e.g.}~\cite{das2010network,ramakrishnan2010network,ho2008network} and references therein). Recent work~\cite{kamath2011generalized,kamath2013study} established 
upper bounds on the multiple unicast sum rate. These bounds either involve edge cut bounds or linear programs involving Shannon inequalities. To the best of our knowledge, these require complexity exponential in the network size to evaluate and it is not known how the gap from achievable schemes can scale. 

We obtain a new upper bound for the sum rate of the optimal vector-linear code for multiple unicast network coding. 
Our bound is established in four steps that are pictorially shown in Fig.\ref{fig:Flowdiag}. 
The first step is bounding the sum rate of a multiple unicast code $R^{MU}$ by the rate of an artificial problem that 
we call correlated unicast coding problem. This problem is a multiple unicast problem that allows an arbitrary correlation between sources but penalizes for joint entropy rate loss. Subsequently, we show that a correlated unicast code is equivalent to a GLRC defined on a suitable recoverability graph. In this equivalence, the joint entropy rate $R^{CO}$ of the correlated unicast code equals the rate of the GLRC. With duality, this is equivalent to linear index coding on the same graph. The last step is to deploy a previous result which showed that complementary index coding can be well-approximated ~\cite{comp index}.  
Our bound can be computed in polynomial time and relies on an approximate cycle packing computed on an index coding instance obtained after transformations. 

We note here that the approximation result for complementary index coding relies on deep results~\cite{seymour1995packing,2004packing} from combinatorial optimization and imply no interesting approximation results for index coding capacity. Our duality result allows us to obtain strong approximation results for GLRC and through our steps for the multiple unicast problem because it maps directly to the complementary index coding problem. 
In this work, we do not rely on the index coding equivalence to general network coding~\cite{effros2012equivalence}.
Further, we emphasize that all our results are valid for linear and vector-linear problems only. 

Finally, we believe that the sum rate of correlated unicasts is equal to the sum rate of multiple unicasts. This, if true, combined with our results would yield an approximation for multiple unicasts within a $\log n \log\log n$ factor which would be a breakthrough. 

	  \begin{figure}
  \centering
   \includegraphics[width=9cm]{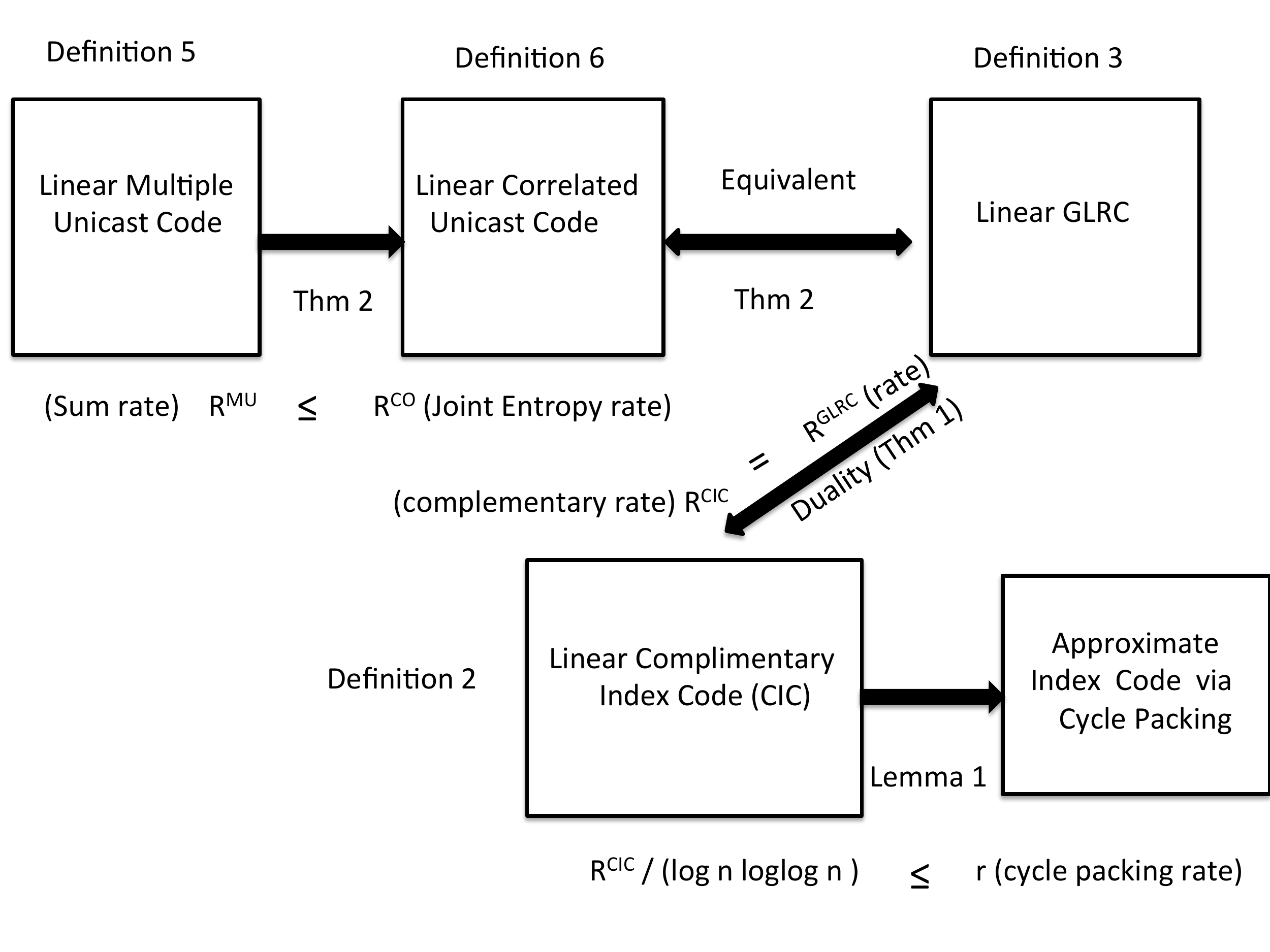}
   \caption{ A pictorial overview of the steps involved for bounding the sum rate of the multiple unicast problem. Theorem  \ref{thm:upperbound} initially shows how the sum rate of a multiple unicast code is upper bounded by the joint entropy of a new code, over the same network, which we call the correlated unicast code.  Theorem \ref{thm:upperbound} shows that a correlated network code is equivalent to a Generalized Locally Repairable Codes (GLRCs) over a suitable recoverability digraph. Theorem \ref{thm:duality} shows the equivalence of GLRC to complementary index coding problem which is subsequently approximated using the work of \cite{comp index}. }
   \label{fig:Flowdiag}
 \end{figure}   

\section{Definitions}

         In this section, we formally define a vector linear Index Code (IC),  a vector linear Generalized Locally Repairable Code (GLRC), a vector linear multiple unicast code and a vector linear correlated unicast code. In this work, we use the terms linear code and subspace interchangeably. In the subsequent sections, we show a duality relation between the first two entities and use it to derive tractable upper bounds on the optimal linear sum rate of the third.
         
         \begin{definition}
                   An \textit{index coding} problem instance is given by $n$ distinct messages, $\mathbf{x}_i ~ 1 \leq i \leq n$ with $\mathbf{x}_i \in \Sigma^{p}$, each intended for a distinct user among a set of $n$ users. Every user has some side-information which is described by a set of indices, $S_i \subseteq \{1,2,3 \ldots n \}$, such that $j \in S_i$ implies that user $i$ has packet $\mathbf{x}_j$ as side information and $i \notin S_i$. This is represented by a directed side information graph $\bar{G} (V,E)$ where each vertex represents a user and a directed edge from $i$ to $j$ is present  if $j \in S_i$.  $\hfill \lozenge$
         \end{definition} 
         
      For ease of notation, let $\mathbf{x}= \left[\mathbf{x}_1^T ~\mathbf{x}_2^T \ldots \mathbf{x}_n^T \right]^T$. The objective is to design suitable transmission schemes such that each user decodes its desired packet from the encoded transmission and the side information packets available with them. Formally, a vector linear index code, which represents a linear transmission scheme, is defined as follows:  
         
         \begin{definition}
             A valid $\left( \Sigma,p,n,k \right)$ \textit{vector linear index code}, for an index coding problem on $\bar{G}(V,E)$, is a collection of $k$ linear encoding vectors $\mathbf{v}_i \in \Sigma^{pn \times 1}$ spanning a subspace ${\cal C} \in \Sigma^{pn}$ of dimension $k$ such that, from the $k$ broadcast transmissions $\mathbf{v}_i^T \mathbf{x}$, all users are able to decode their respective packets using their side-information using linear decoding. In other words, there are decoding functions $\phi_i: \phi_i( \{\mathbf{v}_i^T \mathbf{x} \}_{i=1}^k, \{ \mathbf{x}_j\}_{j \in S_i} ) = \mathbf{x}_i, ~\forall i$ which are linear in all the arguments (in all the subsymbols belonging to $\Sigma$).   $\hfill \lozenge$       
         \end{definition}
         
         The broadcast rate of the index code is given by $k/p $ since every channel use consists of $p$ symbols from the alphabet $\Sigma$. The total number of transmissions is $k$ in terms of the alphabet $\Sigma$. The total number of transmissions that is needed if side information is not present is $np$. The index code ${\cal C}$ has the following generator matrix with the encoding vectors $\mathbf{v}_i$ as the rows.
          \begin{equation}\label{Eqn:genind}
               \mathbf{V}= \left[
                 \begin{array}{c}
                   \mathbf{v}_1^T \\
                   \mathbf{v}_2^T \\
                    \cdot         \\ 
                    \cdot          \\  
                    \mathbf{v}_k^T
                 \end{array}
                \right].
          \end{equation}
    $\mathbf{y}=\mathbf{V}\mathbf{x}$ is the vector containing the $k$ encoded transmissions corresponding to the index code ${\cal C}$. The \textit{complementary index coding} problem is essentially the same as the index coding problem except that the objective is to maximize the number of transmissions saved. The number of saved transmissions is $(np-k)$. The \textit{complementary index coding} rate is given by $\left(n - k/p\right)$ since $\log \left( \Sigma \right)$ bits are transmitted every channel use. Let $R^{CIC} \left( \bar{G} \right)$ be the maximum complementary index coding rate over all the linear codes for the side information graph $\bar{G}$.         
         
         \begin{definition}
             A $\left( \Sigma,p,n,k \right)$ vector linear \textit{generalized locally repairable code} (GLRC) of dimension $k$ is a $k$ dimensional subspace ${\cal C} \subseteq \Sigma^{pn} $ where each set of $p$ subsymbols is grouped into one codeword supersymbol. Further, a codeword supersymbol $i$ satisfies the following recoverability condition: every subsymbol of the $i$th supersymbol is a linear combination of the subsymbols belonging to a set $S_i$ of codeword supersymbols not containing $i$. These conditions can also be represented in the form of a directed recoverability graph $\bar{G}(V,E)$ where the vertices correspond to the $n$ supersymbols and the directed out-neighborhood of a vertex $i$ is the recoverability set $S_i$.  $\hfill \lozenge$
         \end{definition}

	   A GLRC ${\cal C}$ is said to be \textit{valid} on the recoverability digraph $\bar{G}$ if it satisfies the conditions given by the digraph. The generator matrix, of dimensions $k \times pn$, for the code ${\cal C}$ is given by:
	    \begin{equation}\label{Eqn:genGLRC}
	      \mathbf{G}= \left[ \mathbf{g}_{11}~ \mathbf{g}_{12} \ldots \mathbf{g}_{1p} ~ \mathbf{g}_{21} ~ \ldots \mathbf{g}_{np} \right]. 
	    \end{equation} 
	   Here, $\mathbf{g}_{ij} \in \Sigma^{k \times 1},~ 1 \leq i \leq n,~ 1 \leq j \leq p$ is the coding vector that determines the $j$th subsymbol of the supersymbol $i$ in a codeword through a linear combination of $k$ message subsymbols.  Let $\mathbf{u} \in \Sigma^{k \times 1}$ be the message to be encoded using the code ${\cal C}$. The codeword corresponding to this, containing $n$ supersymbols, is generated by $\mathbf{u}^T\mathbf{G}$. The recoverability conditions imply that $\mathbf{g}_{ij} \in \mathrm{span} \left( \left\{ \mathbf{g}_{ab} \right\}_{a \in S_i, 1 \leq b \leq p} \right), ~\forall 1 \leq j \leq p$. The \textit{normalized rate} of the GLRC is given by $k/p$. The maximum normalized rate over all the linear codes for a given recoverability graph $\bar{G}$ is denoted by $R^{GLRC} \left( \bar{G} \right)$. 
	   
	       \begin{figure*}
	    \centering
	       \includegraphics[width=11cm]{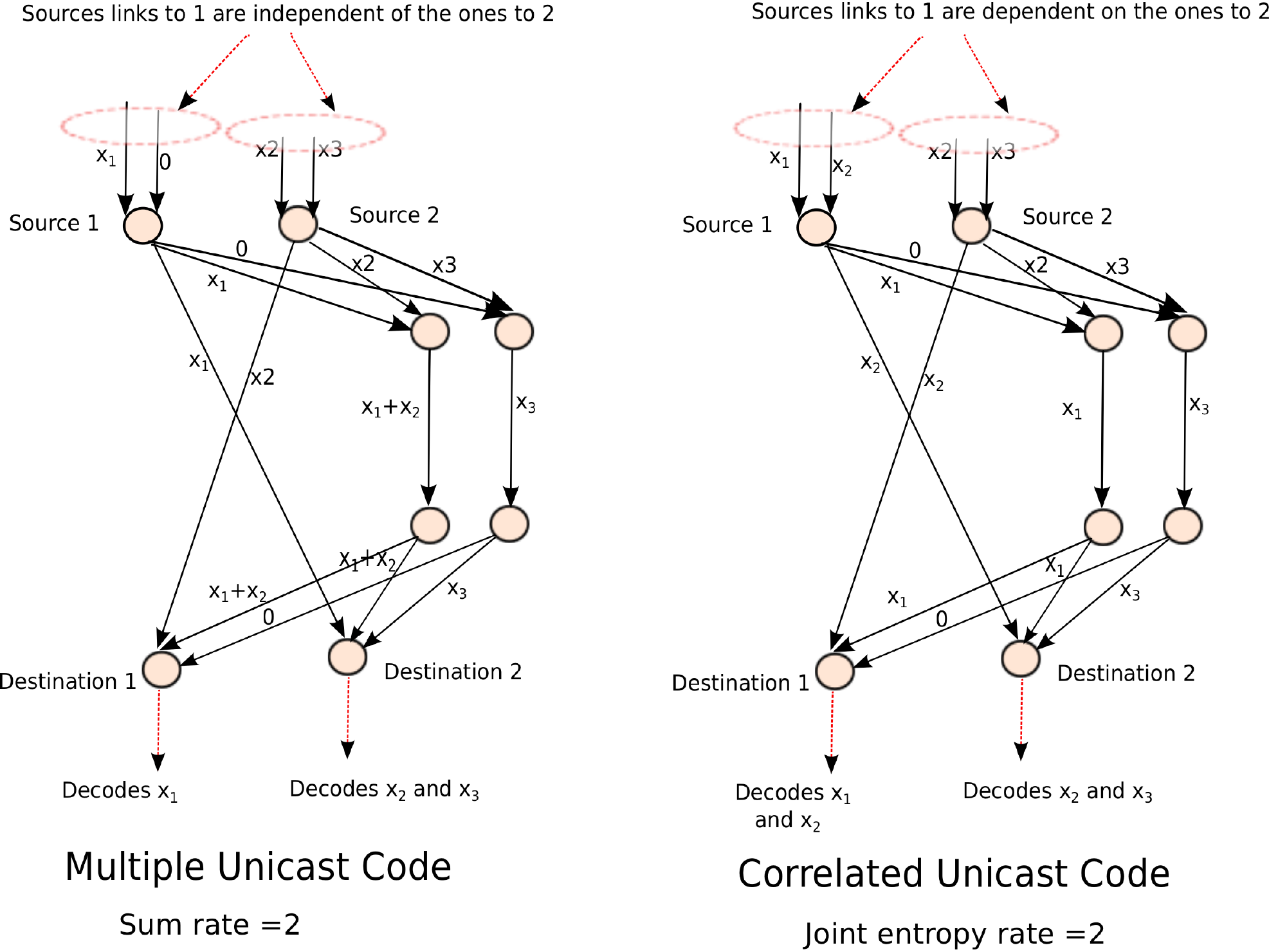}
	     \caption{Illustration of an example multiple unicast network coding instance along with a multiple unicast code and a correlated unicast code. In both cases, destinations decode everything sent along the source links. But the source links across sources are independent in the case of multiple unicast but not so for the correlated unicast code.}
	     \label{fig:network}
	    \end{figure*}

	Now, we provide some definitions regarding the multiple unicast network coding problem.      
	    \begin{definition}
	             A \textit {multiple unicast network coding} instance is given by an acyclic directed network ${\cal G}({\cal N}, {\cal L} )$ that has the following components:
	       \begin{enumerate}      
	             \item ${\cal N}$ is the set of nodes and ${\cal L}$ is the set of directed links each of unit capacity. Unit capacity implies that an edge carries at  most one bit per channel use. A link is denoted by $e$. $h(e)$ denotes the head of edge $e$ and $t(e)$ denotes the tail of edge $e$. Any pair of nodes may have one or more unit capacity links connecting them.
	             \item  (\textit{Source and Destination nodes}) ${\cal S} \subseteq {\cal N}$ is a set of source nodes denoted by $s_1,s_2 \ldots s_r$ where $r= \lvert {\cal S}\rvert$. ${\cal D} \subseteq {\cal N}$ is a set of destination nodes with ${\cal D}= \{d_1,d_2 \ldots d_r \}$.
	             \item (\textit{Source links})  There are source links ${\cal E}_i \subset {\cal L},~1 \leq i \leq r$ such that $h(e)=s_i,~\forall e \in {\cal E}_i$ and these source links do have any tail nodes. They represent information being fed into the network. Further, we place another restriction that $\lvert {\cal E}_i \rvert= \mathrm{mincut} \left(s_i,d_i \right) $. Here, $\mathrm{mincut}(s_i,d_i)$ is the number of edges in the minimum cut between source $i$ and destination $i$.  
	          \end{enumerate}   $\hfill \lozenge$
	    \end{definition}            
	 For ease of notation, let $m=\lvert {\cal L} \rvert$.   
	    \begin{definition}
	          A valid vector linear \textit{multiple unicast network code}, for the network instance ${\cal G}$ with $r$ sources, of dimension $k$ is a subspace ${\cal C} \in \Sigma^{pm}$. A group of $p$ symbols is grouped into a supersymbol and there is a supersymbol for every link $e \in {\cal L}$. Let ${\cal L}= \{e_1,e_2 \ldots e_m \}$. Let $\mathbf{z} \in {\cal C}$ and $\mathbf{z}_e \in \Sigma^{p \times 1}$ represent a supersymbol (a vector of $p$ subsymbols) corresponding to edge $e$. Then, $\mathbf{z}= \left[ \left(\mathbf{z}_{e_1} \right)^T \ldots \left(\mathbf{z}_{e_m} \right)^T \right]$ is the vector of all supersymbols. Let $\mathbf{G} \in \Sigma^{k \times mp}$ be the generator matrix of the code comprising columns $\mathbf{g}_{ej} \in \Sigma^{k \times 1}$ for all links $e$ and subsymbols $j$ for $1 \leq j \leq p$. Given a $k \times 1$ message vector $\mathbf{x}$, $\mathbf{x}^T \mathbf{G} =\mathbf{z}$ where $z_{ej}$ is the $j$ th subsymbol on link $e$. Further, they satisfy the following criteria:
	          \begin{enumerate}
	              \item (\textit{Coding at intermediate nodes}) There exists $\phi_e:\mathbf{z}_e = \phi_e \left( \{ \mathbf{z}_a \}_{a:t(a)=h(e)} \right)$ where $\phi_e$ is the local vector linear encoding function at an edge such that every information subsymbol on that edge is a linear combination of all subsymbols arriving at its head.  
	              \item (\textit{Decoding at destinations}) For every source $i$ and $~\forall e \in {\cal E}_i$, $\mathbf{z}_e = \hat{\phi}_e^{i} \left( \{ \mathbf{z}_a \}_{a:t(a)=d_i} \right) $. Here, $\hat{\phi}_e^{i}$ is a vector linear decoding function such that every information subsymbol on a source source link $e$ is decoded by a linear combination of all the subsymbols arriving at its corresponding destination.  
	              \item (\textit{Independence between sources} ) Information arriving at source $i$ through the source links is independent of the information arriving at source $j \neq i$ through its source links.  Formally, in terms of the generator columns, $\mathrm{span} \left( \{ \mathbf{g}_{eb}\}_{e \in {\cal E}_i, 1 \leq b \leq p} \right) \cap \mathrm{span} \left( \{ \mathbf{g}_{eb}\}_{e \in {\cal E}_j, 1 \leq b \leq p} \right) = \emptyset$ for $i \neq j$.	             	              
	          \end{enumerate} $\hfill \lozenge$
	    \end{definition}
	    The entropy of source $i$ is given by $\mathrm{dim} \left( \{ \mathbf{g}_{eb}\}_{e \in {\cal E}_j, 1 \leq b \leq p} \right) \log (\Sigma)$ bits. The joint entropy of all the sources is $k \log (\Sigma)$ bits. In this work, we would not be concerned about the individual source entropies. Because of independence between sources, the sum rate supported by the network equals the joint entropy rate of all sources. The sum rate is given by $k/p$ bits/network use since the network is used $p$ times and each use carries a symbol from the alphabet $\Sigma$. Let $R_{MU} \left({\cal G} \right)$ denote the maximum vector linear sum rate supported by the network ${\cal G}$ over all possible multiple unicast linear codes ${\cal C}$.
	    
	     In the network coding theory parlance, sometimes the columns $\mathbf{g}_{eb}$ are called global encoding functions. The local encoding/decoding functions $\left( \phi_e, \hat{\phi}_e^i \right)$ can be obtained (may not be unique) from global encoding functions $\mathbf{g}_{ej}$.  
	    
	     For the purposes of obtaining bounds on $R_{MU} \left({\cal G} \right)$, we define a \textit{correlated unicast} code on the network ${\cal G}$. 
	    \begin{definition}
	          A valid vector linear \textit{correlated unicast code}, on the multiple unicast network instance ${\cal G}$ with $r$ sources, of dimension $k$ is a subspace ${\cal C} \in \Sigma^{pm}$ whose definition is identical to the multiple unicast code except that the last criterion of independence between sources is not imposed. In other words, the subspaces spanned by the generator columns corresponding to the source links for different sources can overlap. $\hfill \lozenge$	          
	     \end{definition}
	     
	  In this case, the joint entropy of all the sources is still given by $k \log (\Sigma)$ bits and the \textit{joint entropy rate} (and not the sum rate) is given by $k/p$ bits per network use. Let $R^{CO} \left( {\cal G}\right)$ denote the optimum joint entropy rate over all correlated unicast codes supported by the network ${\cal G}$.  
	    
	    \textbf{Remark:} A correlated unicast code is not a network code for a multiple unicast correlated sources problem (similar to the multicast model in \cite{ho2004network}) because existence of a correlated unicast code with joint entropy rate $h$ for a unicast instance ${\cal G}$ implies that sources possessing joint entropy rate $h$ exist with \textit{some} correlation among them which can be transmitted through this network code. The sources are allowed to have arbitrary correlation depending on the code.
	    	    
	    As an illustration of the definitions, a multiple unicast network is provided in Fig. \ref{fig:network}. The network is a $2$ unicast network. The mincut between source $i$ and destination $i$ is $2$. Therefore, each source has $2$ source links. Every link has capacity $1$. A multiple unicast code on this network is also provided. In this, the first source sends $x_1$ and the second sources sends $x_2$ and $x_3$. The sources satisfy the independence condition. But for the correlated unicast code, the source links of both sources are correlated. But, that also achieves the joint entropy rate of $2$. Note, that in both cases, every destination decodes whatever the source links carry whether they are correlated with other sources or not.

	  \section{Duality between GLRC and Index Coding}              
	          The main duality result between a GLRC and an Index Code is given by the following theorem:
	          
	      \begin{theorem} \label{thm:duality}   
	             Let ${\cal C}$ be a linear code (or a subspace) of dimension $k$. Let the dual code (or the dual subspace) of ${\cal C}$ of dimension $np-k$ be denoted by ${\cal C}^{\perp} \in \Sigma^{pn}$. Then, ${\cal C}$ is a valid index code for the side information graph $\bar{G}$ iff ${\cal C}^{\perp}$ is a valid GLRC when $\bar{G}$ is taken as a recoverability graph.  $\hfill \lozenge$
            \end{theorem}          
             \begin{proof}
                    We first show that if ${\cal C}$ is a valid index code on $\bar{G}$, with generator $\mathbf{V}$ as in (\ref{Eqn:genind}), then the dual code ${\cal C}^{\perp}$ with its generator $\mathbf{G}$ is a valid GLRC code for $\bar{G}$. Consider any user $i$ in the index coding problem. Let the side information set be $S_i$. If ${\cal C}$ is a valid index code, then there exists a vector linear (linear in all the subsymbols)  decoding function $\phi_i: \phi_i \left(\mathbf{y}, \{ \mathbf{x}_j \}_{j \in S_i} \right) = \mathbf{x}_i$. This is true for all message vectors $\mathbf{x}: \mathbf{y}= \mathbf{V}\mathbf{x}$. Let $\mathbf{w}$ be a vector such that $\mathbf{y}=\mathbf{V}\mathbf{w}$. Let $\mathbf{x}$ represent the actual message vector (of all $n$ messages). Let the encoded transmission be $\mathbf{y}$. Then, $\mathbf{x}=\mathbf{w}+ \mathbf{z}$ for some $\mathbf{z} \in {\cal C}^{\perp}$ because ${\cal C}^{\perp }$ is the right null space of $\mathbf{V}$.
 
 Given $\mathbf{y}$, the uncertainty about message vector $\mathbf{x}$ is because of the unknown $\mathbf{z}$ in the null space. In that sense, given the generator $\mathbf{V}$ of the code, one can fix a candidate $\mathbf{w}$ for a given $\mathbf{y}$. Because $\phi_i$ is linear in all the arguments, we have the following chain of inequalities:
                \begin{align} 
                        \hfill &\phi_i (\mathbf{y}, \{\mathbf{x}_j \}_{j \in S_i})  = \mathbf{x}_i \nonumber \\
                         \Rightarrow & \phi_i (\mathbf{y}, \{\mathbf{w}_j +\mathbf{z}_j \}_{j \in S_i})  = \mathbf{w}_i + \mathbf{z}_i \nonumber \\
                         \Rightarrow & \phi_i (\mathbf{y}, \{\mathbf{w}_j \}_{j \in S_i} ) + \phi_i (\mathbf{0}, \{\mathbf{z}_j \}_{j \in S_i} ) = \mathbf{w}_i + \mathbf{z}_i 
                         \label{Eqn:chainind}
                \end{align}
            The last step uses linearity of $\phi_i$. The decoding should work even when $\mathbf{w}$ is the actual message vector. Hence, $\phi_i (\mathbf{y}, \{\mathbf{w}_j \}_{j \in S_i} )= \mathbf{w}_i $. With $(\ref{Eqn:chainind})$, we have: 
                     \begin{equation}
                        \phi_i \left( \mathbf{0}, \{\mathbf{z}_j \}_{j \in S_i} \right) = \mathbf{z}_i
                     \end{equation}   
              Since $\phi_i$ is linear, this implies that every subsymbol of the $i$th code supersymbol is linearly dependent on all the code subsymbols in the set $S_j$ for the dual code ${\cal C}^{\perp}$ since $\mathbf{z} \in {\cal C}^{\perp}$. Hence, the dual code is a valid GLRC proving one direction.
               
                To prove the other direction, let us assume that for every $i:1\leq i \leq n$, there exist functions $\tilde{\phi}_i$ such that :
            \begin{equation} \label{eqn:recovcond}
             \tilde{\phi}_i \left( \{ \mathbf{z}_j \}_{j \in S_i} \right) = \mathbf{z}_i, ~\forall \mathbf{z} \in {\cal C}^{\perp}
              \end{equation}
              
                 Here, $\mathbf{z}$ is a vector of all supersymbols $\mathbf{z}_i$. This means that every supersymbol $i$ of the GLRC code ${\cal C}^{\perp}$ is recoverable from the set $S_i$ of codeword supersymbols.  For the index coding problem, let $\mathbf{x}$ be the message vector not known to the users prior to receiving the encoded transmission. Let $\mathbf{y} = \mathbf{V}\mathbf{x}$. Given $\mathbf{y}$, from the previous part of the proof, we know that $\mathbf{y}= \mathbf{w}+\mathbf{z}$ for some $\mathbf{z} \in {\cal C}^{\perp}$. $\mathbf{w}$ is known to all users from just $\mathbf{y}$ because the code $\mathbf{V}$ employed is known to all the users.
                 
                  Since $\mathbf{z}$ satisfies the recoverability conditions in (\ref{eqn:recovcond}), $\mathbf{w}_i+\tilde{\phi}_i \left( \{ \mathbf{x}_j -\mathbf{w}_j \}_{j \in S_i} \right) = \mathbf{x}_i$. $\mathbf{w}$ is a function of just $\mathbf{y}$ and $\mathbf{V}$.  Hence, user $i$ can recover $\mathbf{x}_i$ from supersymbols from the side information set $S_i$ and the encoded transmission $\mathbf{y}$ for all message vectors $\mathbf{x}$. 
                  
                  We again note that the choice of $\mathbf{w}$ is arbitrary. For every $\mathbf{y}$, users have to pick some $\mathbf{w}$ such that $\mathbf{y}=\mathbf{V}\mathbf{w}$. Since the forward map is linear, the inverse one-to-one map  ${\cal V}^{-1} \left(\mathbf{y}\right)$ determining $\mathbf{w}$ can be made linear by fixing ${\cal V}^{-1} \left( \mathbf{e}_i \right)$ for all unit vectors $\mathbf{e}_1, \ldots \mathbf{e}_k$. Then, linearity of the forward map determines a candidate pre-image for all vectors $\mathbf{y}$, i.e. ${\cal V}^{-1} \left( \mathbf{y} \right) = \sum \limits_{i=1}^k y_i {\cal V}^{-1} \left( \mathbf{e}_i \right) $. Therefore, if $\tilde{\phi}_i$ are all linear in all the subsymbol arguments, then the decoding functions for the index coding problems are also linear.  This completes the proof.
                          
             \end{proof}
        \textbf{Remark:} In the above proof, for the forward direction, we assumed linearity of decoding functions for both the index code and GLRC. For the reverse part, the arguments were more general even admitting non linear decoding functions. Our definitions for the linear index code and GLRC involves only linear decoding. We note that the proof of Theorem $1$ in \cite{bar2011index} implies that for linear index codes, linear decoding is optimal. Although the argument provided in \cite{bar2011index} is only for the scalar binary case, the same argument can be extended to vector linear codes over any field. This, with the above proof, implies that even for a linear GLRC, linear recoverability functions are sufficient for recovery. Altogether, there is no loss of generality in the definitions in this work with respect to decoding. 
                      
             \section{Bounds on the linear multiple unicast sum rate}
            In this section, we derive a polynomial time computable upper bound for $R^{MU} \left( {\cal G}\right)$ which is within $\log \left( \lvert {\cal L} \rvert \right) \log\log\left( \lvert {\cal L} \rvert \right)$ from $R^{CO} \left( {\cal G} \right)$ where ${\cal G}$ is a multiple unicast network instance and ${\cal L}$ is the set of links as defined in the previous section. First, we show that $R^{GLRC} \left( \bar{G} \right)$ can be approximated within a factor of $\log(n) \log\log(n)$ in polynomial time where $\bar{G}$ is a directed graph on $n$ vertices. This uses the duality result in the previous section and existing results in approximating the complementary index coding problem.
             \begin{lemma}\label{lem:approx}
              A valid GLRC on a digraph $\bar{G}$ with rate $r$ satisfying $r \geq \frac{R^{GLRC} \left( \bar{G}\right)}{\log n \log\log (n) }$ can be computed in polynomial time.
              \end{lemma}                                  
             \begin{proof}
                The duality result of Theorem \ref{thm:duality} means that a vector linear GLRC has normalized rate $r$ on $\bar{G}$ iff there is a feasible vector linear index code on $\bar{G}$ with complementary index coding rate $r$ such that one code is a dual of the other. Further, it has been shown in \cite{compindex} that the algorithm from \cite{2004packing}, used to find a fractional cycle packing of $\bar{G}$, along with integrality gap results on the feedback edge set problem from \cite{seymour1995packing} yields a vector linear binary index code, whose complementary index coding rate is $r \geq \frac{R^{CIC} \left( \bar{G} \right)}{\log(n) \log \log (n)}$, in polynomial time. By the duality result, we also have $R^{CIC} \left( \bar{G} \right)=R^{GLRC} \left( \bar{G} \right)$.  Further, the dual of the binary vector linear index code obtained is also a valid GLRC with the same rate $r$. Hence, the result follows.
             \end{proof}
             Now, we use this to upper bound $R^{MU} \left( \bar{G} \right)$ in the following theorem which is the main result in this paper.
             
             \begin{theorem} \label{thm:upperbound}
                  $r$ is a polynomially computable function of a multiple unicast network ${\cal G}$ such that $R^{MU} \left( {\cal G} \right) \leq R^{CO} \left( {\cal G} \right) \leq r \log \left( \lvert {\cal L}\rvert \right) \log \log \left( \lvert {\cal L}\rvert \right) \leq R^{CO} \left( {\cal G} \right) \log \left( \lvert {\cal L}\rvert \right) \log \log \left( \lvert {\cal L}\rvert \right) $ where ${\cal L}$ is the set of links.
             \end{theorem}
              \begin{proof}
                 $R^{MU} \leq R^{CO}$ is clear because, for the multiple unicast code, there is just an added restriction of requiring independence between sources over the correlated unicast code.
                                                   
                 For the other parts, we show that a correlated unicast code on the network $\left( {\cal G} \right)$ is identical to a GLRC on a digraph $\bar{G}\left(V,E \right)$ which we construct as follows: There is a node for every edge in the network, i.e. ${V= \cal L}$. If the edge $e$ is not a source edge, define the recoverability set $S_e= \{ e' \in {\cal L}: h(e')=t(e) \}$. If $e \in {\cal E}_i$ (a source edge feeding into source $i$) for some $i$, then $S_e= \{e' \in {\cal L}: h(e')=d_i \}$. Recoverability set $S_e$ forms the directed out-neighborhood of vertex $e$ in $\bar{G}$. In other words, $(e,e') \in E$ iff $e' \in S_e$. It is easy to see that a GLRC code for $\bar{G}$ of dimension $k$ is exactly the same as a correlated unicast code for ${\cal N}$ of dimension $k$ and vice versa. This is because the decodability conditions at the destinations and local encoding conditions translate to recoverability conditions for the GLRC and vice versa.   Hence, $R^{CO} \left( {\cal G} \right)= R^{GLRC} \left( \bar{G}\right)$.              
                 
                 From Lemma \ref{lem:approx}, we know that there is a polynomial time computable function $r$ which is the rate of a feasible GLRC code (this code can also be obtained) such that $r \leq R^{GLRC} \left( \bar{G} \right) \leq r \log \left( n \right) \log \log (n)$. Since, $R^{GLRC} \left( \bar{G} \right) = R^{CO} \left({\cal G}\right) $, the result in the theorem follows.  
              \end{proof}
              
            We make an important observation: the absence of the condition describing independence of sources in the definition of the correlated unicast code is the prime reason for the equivalence between GLRC and the correlated unicast code.  We observe that the recoverability conditions for GLRC is a 'list of linear dependencies' among a system of vectors.  According to the interference alignment interpretation of \cite{maleki2012index}, index coding is a 'list of linear independencies'. In a very rough sense, we have shown that a list of linear dependencies is the 'dual' of a related list of linear independencies.
We were not able to show a complete equivalence between multiple unicast network code and GLRC because of the condition requiring independence of sources which directly cannot be written as a dependency condition.           
             
             For the multiple unicast network ${\cal G}$, we have taken the number of source links entering source $i$ to be exactly equal to the $\mathrm{mincut}(s_i,d_i)$. If the number of source links in the definition is increased beyond mincut, $R^{MU}$ would not be affected. This is because the rate supported by every source is bounded by the mincut between that source and the destination. However, for the correlated unicast code, increasing the number of source links beyond mincut can increase the joint entropy rate of the correlated unicast code beyond that of the multiple unicast sum rate. It is possible to find such examples. However, with the present definition, where the number of source links is exactly equal to mincut, we have not been able to find a network where $R^{MU} < R^{GLRC}$. 
             
             We believe that, for a multiple unicast network ${\cal G}$ where the number of source links equals the mincut between the corresponding source and destination, i.e. $\lvert {\cal E}_i \rvert = \mathrm{mincut} \left(s_i,d_i\right)$, $R^{MU} \left( {\cal G} \right) = R^{CO} \left( {\cal G} \right)$.
             
             Proving this would mean that the computable function $r$ in Theorem \ref{thm:upperbound} is also a lower bound to the multiple unicast linear sum rate, i.e. $r \leq R^{MU} \left( {\cal G} \right)$. This would mean tractably approximating $R^{MU}$ within a $\log (\lvert {\cal L} \rvert) \log\log (\lvert {\cal L} \rvert)$ factor. Even if this does not exactly hold, finding how these are related could give new lower bounds on the multiple unicast sum rate.
             \section{Conclusion}
                   We showed a duality between Index Codes and Generalized Locally Repairable Codes (GLRCs). Further, approximation algorithms on the complimentary index coding problem together with this duality result give a polynomial time computable upper bound to the optimum linear sum rate of the multiple unicast problem. At the heart of these results, lies the usage of a correlated unicast code that relates GLRC to the multiple unicast problem. Any progress on determining the relationship between the correlated unicast code and the multiple unicast code would lead to a tractable way of approximating the linear sum rate of the multiple unicast problem. 
             
\pagenumbering{arabic}
\bibliographystyle{IEEEtran}
\bibliography{GLRCbibv2}

\end{document}